\documentclass[10pt,twocolumn,twoside]{IEEEtran}
\usepackage{graphicx,graphics,latexsym,verbatim,epsfig,subfigure}
\usepackage{mathrsfs}
\usepackage{amssymb,amsmath}
\usepackage{amsthm}
\usepackage{tikz}
\newtheorem{lemma}{Lemma}

\newtheorem{definition}{Definition}

\newtheorem{theorem}{Theorem}
\newtheorem{corollary}{Corollary}
\newtheorem{remark}{Remark}
\newtheorem{assumption}{Assumption}
\newtheorem{proposition}{Proposition}

\IEEEoverridecommandlockouts                              

\title{\LARGE \bf
Dissipativity Tools for Convergence to Nash Equilibria in Population Games
}
\author{Murat Arcak$^{1}$ and Nuno C. Martins$^{2}$
\thanks{*Arcak's work was supported in part by the National Science Foundation grant CNS-1545116. Work by Martins was supported by AFOSR grant FA9550-19-1-0315.}
\thanks{$^{1}$ Murat Arcak is with the EECS Dept., University of California, Berkeley        {\tt\small arcak@berkeley.edu}}%
\thanks{$^{2}$ Nuno Miguel Lara Cintra Martins is with the ECE Dept. and ISR at the University of Maryland, College Park
        {\tt\small nmartins@umd.edu}}%
}

\begin{document}

\maketitle

\begin{abstract}
We analyze the stability of a nonlinear dynamical model describing the noncooperative strategic interactions among the agents of a finite collection of populations. Each agent selects one strategy at a time and revises it repeatedly according to a protocol that typically prioritizes strategies whose payoffs are either higher than that of the current strategy or exceed the population average. 
The model is predicated on well-established research in population  and evolutionary games, and has two components.
The first is the payoff dynamics model (PDM), which ascribes the payoff to each strategy 
according to the so-called social state vector whose entries are the proportions of every population adopting the available strategies. The second component is the evolutionary dynamics model (EDM) that accounts for the revision process. In our model, the social state at equilibrium is a best response to the strategies' payoffs, and  can be viewed as a Nash-like solution that has predictive value when it is globally asymptotically stable (GAS). We present a systematic methodology that ascertains GAS by checking separately whether the EDM and PDM satisfy appropriately defined system-theoretic dissipativity properties. Our work generalizes pioneering methods based on notions of contractivity applicable to memoryless PDMs, and more general system-theoretic passivity conditions. As demonstrated with examples, the added flexibility afforded by our approach is particularly useful when the contraction properties of the PDM are unequal across populations.
\end{abstract}


\vspace{-.1in}
\section{Introduction}
Consider a large number of agents that interact by selecting strategies in a noncooperative way. Each agent selects one strategy at a time and revises it repeatedly in response to a payoff vector whose entries are the strategies' payoffs. The revision process of each agent is governed by a protocol that, in general, is probabilistic and prioritizes strategies with a higher payoff.
Each agent belongs to one  of a finite collection of populations, and agents in the same population choose strategies from a common finite set, follow an identical revision protocol, and access the same payoff vector. Although each population is protocol-homogenous, its agents are allowed to simultaneously select distinct strategies. The state of a population is a vector whose entries are proportional to the number of agents adopting each strategy. By concatenating the states of all populations, we generate the so-called social state. A causal payoff mechanism determines the payoff vectors of the populations in terms of the social state. Hence, although the agents are noncooperative, their decisions are coupled through the payoff mechanism.

\subsection{Motivation and mixed autonomy congestion~games} \label{subsec:motivation}
Our framework is well-suited to model multi-agent systems in which a large number of agents use the available information, such as the payoff vectors, to autonomously choose and repeatedly revise their strategies according to protocols that express their preferences. In the so-called congestion game originally proposed in~\cite{Beckmann1956Studies-in-the-}, and thoroughly explained in~\cite{Sandholm2010Population-Game} (see also~\cite[Example~2]{Park2019From-population} and~\cite[Example~1]{Park2018Payoff-Dynamic-}), the agents are the drivers commuting across a network of roads. Each population is uniquely tied to an origin-destination pair shared by all its members, and the strategies available to them are the viable routes leading from the origin to the destination. The time saved by adopting a route is its payoff\footnote{The antisymmetric of  travel time from  origin to  destination is a mathematically equivalent representation 
of a strategy's payoff.
}, and the proportions of the agents in each population adopting the available routes form the population states. Under the realistic assumption that as the utilization of a road (measured as the proportion of all agents using it) increases so does the time of travel across it, congestion games satisfy important properties that facilitate their analysis. Namely, congestion games are (i)~potential~\cite{Monderer1996Potential-games} and (ii)~contractive~\cite{Hofbauer2007Stable-games,Hofbauer2009Stable-games-an}, which are properties that have been used to prove the existence of Nash equilibria and construct Lyapunov functions to establish the stability of the equilibria when the agents follow suitable protocols, such as Smith's originally proposed in~\cite{Smith1984The-stability-o} to study traffic assignment problems. 

Throughout the article we  illustrate  our results via two examples. The first is a generalization of congestion games in which, in addition to drivers, there are autonomous vehicles~\cite{LazCooPed17,MehHor19}. As we explain later, existing stability results are not applicable to such mixed autonomy congestion games because they are neither potential nor contractive. Our results, however, allow for what we will define as {\it weighted contractive} games, of which mixed autonomy congestion games are a particular case. 
In the presence of dynamics in the payoff mechanism, standard methods characterize the stability of the equilibria of the game with the help of a potential function. Since no such  potential function exists for mixed autonomy congestion games, we instead use a technique based on Legendre's transform.
In the second example, we  illustrate how our results for weighted contractive games can be applied to study bypassing behaviour near a road split.


The modeling and analysis tools such as those put forth in this article are critical for understanding user behavior and for devising policies for efficient sharing of resources, such as infrastructure. The importance of such studies will increase as the users are diversified (as in the mixed autonomy example) and real-time information that guide their choices proliferate.

\vspace{-.1in}
\subsection{Main Contributions and Comparison To Previous Work}
In this article, we generalize  previous  results characterizing the convergence of the social state towards an appropriately defined Nash-like equilibrium set of the payoff mechanism. 
As we explain in Section~\ref{subsec:deterministicmodels}, stability of the equilibria set is a critical property because it ensures that such a set is a predictor of the long term behavior of the social state.

For the case in which the payoff mechanism is a memoryless map~(from the social state to the payoff vectors) derived from a contractive game\footnote{Although population games satisfying this property were originally called {\it stable} in~\cite{Hofbauer2009Stable-games-an}, we  refer to them as {\it contractive} following the nomenclature in~\cite{Sandholm2015Handbook-of-gam}. This convention is  appropriate because, as illustrated by~\cite[Example~6.1]{Hofbauer2009Stable-games-an}, contractive games may still lead to cycles that exclude Nash equilibria for certain protocols. Furthermore, \cite[Section~V.B~(Fig.1(a))]{Park2018Passivity-and-e} demonstrates that contractivity is also not necessary for GAS of the Nash equilibria.}, seminal work in~\cite{Hofbauer2009Stable-games-an} puts forth a Lyapunov-based approach to establish the global asymptotic stability (GAS) of the set of Nash equilibria of the game for a broad class of revision protocols.  
 
By using system-theoretic passivity concepts, as  introduced in this context in~\cite{Fox2013Population-Game}, subsequent work in~\cite{Park2018Passivity-and-e,Park2019From-population,Park2018Payoff-Dynamic-} proposed a methodology that generalizes the stability results in~\cite{Hofbauer2009Stable-games-an} for dynamical payoff mechanisms. 
 This generalization is important because the {addition of dynamics} in the payoff mechanism {may destabilize} the Nash equilibria of a contractive game for protocols that  would have guaranteed stability in the memoryless case. In fact, even first order {\it smoothing dynamics}~(that smooth short-term fluctuations of the payoff)
 may cause such a destabilization effect. An example is provided in~\cite[Section V.B (Fig.5)]{Park2019From-population} where the revision protocol yields GAS of the Nash equilibria of a contractive game with a memoryless payoff mechanism, but stability is lost when smoothing dynamics are added to the payoff mechanism.

In this paper, we present dissipativity tools that further advance the passivity approach used in~\cite{Fox2013Population-Game,Park2018Passivity-and-e,Park2019From-population,Park2018Payoff-Dynamic-}, and allow us to establish GAS of the set of Nash equilibria for broader classes of payoff mechanisms. There are two sets of results, centered on (i) memoryless payoff mechanisms and (ii) the case in which the payoff mechanism has internal dynamics. Our results for memoryless payoff mechanisms focus on those derived from {\it weighted contractive} games, whose contractiveness properties may differ from one population to another.
The game of congestion with mixed autonomy~\cite{LazCooPed17,MehHor19} alluded to in Section~\ref{subsec:motivation} is weighted contractive and will be used as an example throughout the article. We employ our results to study the stability of the Nash equilibria of such games for a broad class of protocols, and we also consider the case in which the payoff mechanism is modified to include smoothing dynamics. In addition, we propose a numerical method that leverages convex optimization to determine whether a game satisfies the relaxed contraction properties that are consistent with the generalized dissipativity properties.

\vspace{-.05in}
\subsection{Deterministic Models: A Brief Discussion}
\label{subsec:deterministicmodels}

We adopt the deterministic continuous time dynamics model described in~\cite{Park2019From-population,Park2018Payoff-Dynamic-}, which focuses on the {\it mean closed loop model} depicted in Fig.~\ref{Fig:ClosedLoop}. It consists of the feedback interconnection of two nonlinear sub-systems: the first is an evolutionary dynamics model (EDM) that models the effect of the revision protocols, and the second is a payoff dynamics model (PDM) that specifies the payoff mechanism. The state of the EDM and the PDM are the so-called mean social state and deterministic payoff, respectively, that approximate the social state and the payoff vectors when the number of agents tends to infinity as described in~\cite{Park2019From-population,Park2018Payoff-Dynamic-}. Our model assumes that the protocols satisfy the so-called Nash stationarity property, which ensures that the mean social state components of the closed loop equilibria 
coincide with the Nash-like equilibria of the PDM. As is discussed in detail in~\cite[Section~V]{Park2018Payoff-Dynamic-}, the analyses in~\cite{Kurtz1970Solutions-of-Or} and~\cite[Appendix~12.B]{Sandholm2010Population-Game} indicate that the convergence of the social state towards Nash-like equilibria, in the limit of large populations, can be established by doing so for the mean social state. For these reasons, in this article we study the mean closed loop model and leverage dissipativity theory to determine conditions under which the equilibrium set is~GAS. 

\tikzstyle{system} = [draw, ultra thick, minimum width=6em, text centered, rounded corners, fill=black!4,  minimum height=3em]

\def\edgedist{1.1}
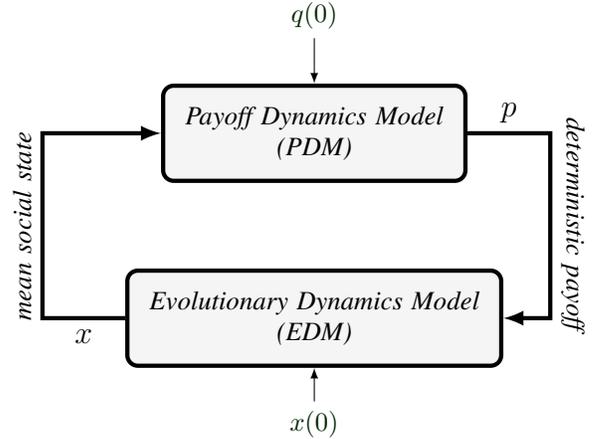
\begin{figure} [t]
\centering
\begin{tikzpicture}[auto, node distance=7em,>=latex]

\node [system, minimum height=8ex] (games) {$\begin{array} {c} \textit{Payoff Dynamics Model} \\ \textit{(PDM)} \end{array}$};
\draw [->] (games)+(0,8ex) node[above,color=green!15!black]{$q(0)$} -- (games);

\node [system, minimum height=8ex, below of=games] (evo_dynamics) {$\begin{array} {c} \textit{Evolutionary Dynamics Model} \\ \textit{(EDM)} \end{array}$};
\draw [->] (evo_dynamics)+(0,-7ex) node[below,color=green!15!black]{$x(0)$} -- (evo_dynamics);

\draw[ultra thick, ->] (games.east) -- ++(1*\edgedist,0) node [pos=.5] {\large $p$} |- (evo_dynamics.east) node [above, pos=.25, rotate=-90] {\textit{deterministic payoff}};
\draw[ultra thick,->] (evo_dynamics.west) -- ++(-1*\edgedist,0) node [pos=.5] {\large $x$} |- (games.west) node [above, pos=.25, rotate=90] {\textit{mean social state}} ;

\end{tikzpicture}
\caption{Diagram representing a feedback interconnection between a PDM and an EDM. The resulting system is referred to as {mean closed loop model}.}
\label{Fig:ClosedLoop}
\end{figure}

The deterministic approach adopted here builds on the extensive body of work on population games and evolutionary dynamics thoroughly discussed in~\cite{Sandholm2010Population-Game,Weibull1995Evolutionary-ga}.
Application of this approach have been reported in the areas of transportation~\cite{Smith1984The-stability-o}, wireless networks~\cite{Tembine2008Evolutionary-ga,Tembine2010Evolutionary-ga,Jian2013Joint-Spectrum-}, optimization~\cite{Li2013Desigining-game}, control systems~\cite{Quijano2017The-role-of-pop}, regulation of heating in buildings~\cite{Obando2014Building-temper}, and smart grid~\cite{Mojica-Nava2014Dynamic-Populat,Pantoja2011A-population-dy}. The analysis in~\cite{Fox2013Population-Game} introduced the concept of dynamic payoff mechanisms in this context, and it also pioneered the use of system-theoretic passivity techniques to characterize stability properties.

\vspace{-.1in}
\subsection{Outline of the Paper} Section~\ref{Sec:1} considers the case where the payoff is memoryless and  defines the basic components of the model.
Section \ref{sec:delta-dissipative}
 introduces the concept of $\delta$-disipativity, which  is then used in Theorem~\ref{main} to specify sufficient conditions for stability.  Section~\ref{weighted} introduces the notion of {\it weighted contractive} games, and shows that these satisfy the conditions of Theorem~\ref{main}. It also proves that a congestion game with mixed autonomy~\cite{LazCooPed17,MehHor19} used throughout the article as an example is weighted contractive.  Section~\ref{sec:Search}  proposes a numerical method that leverages convex optimization to determine whether the payoff satisfies the conditions of Theorem~\ref{main}.
As an illustration
 we check contraction for a game that captures bypassing near a road split  \cite{mehr2018game}.  Section~\ref{sec:DynamicPayoff} extends the  concepts introduced in Section~\ref{Sec:1} to the case in which a payoff dynamics model (PDM) governs the payoff mechanism.  Theorem~\ref{PDMthm} of Section~\ref{sec:DynamicPayoff} replaces Theorem~\ref{main} of Section~\ref{Sec:1} when the payoff mechanism is a PDM instead of being memoryless. Section~\ref{sec:DynamicPayoff} ends by using Theorem~\ref{PDMthm} to characterize  stability for the congestion game with mixed autonomy modified so as to include smoothing dynamics.
 
\vspace{-.05in}
\section{Model Description}\label{Sec:1}


We start by describing concepts used throughout the paper, and defining the elements of the closed loop model. 
In this section we assume the payoff mechanism is {memoryless} and specified by a population game. In Section~\ref{sec:DynamicPayoff}, we extend these results to the case in which the payoff mechanism is a more general~PDM specified by a nonlinear dynamical system. 

We consider $\rho$ populations labeled $\{1,\dots,\rho\}$ and denote the mean state of the $r$th population at time $t$ as
 $x^r(t)\in \mathbb{X}^r$. Here, $\mathbb{X}^r$ is the set of all possible states of population $r$, also commonly called strategy profiles, defined as:
$$ \mathbb{X}^r : =\{ \xi \in \mathbb{R}_{\ge 0}^{n^r}: \ \sum_{i=1}^{n^r}\xi_i=m^r\}, $$
where $m^r$ is a positive constant representing the total ``mass" of the population. The set of strategies available to population $r$ is $\{1,\ldots,n^r \}$ and $x_i^r(t)$ is the mean portion of the population adopting the $i$th strategy at time $t$.
 We define $n:=n^1+\cdots+n^\rho$ and let  $x(t)$ represent the mean social state obtained as the concatenation of the vectors $x^1(t),\dots,x^\rho(t)$. The set $\mathbb{X}$ of all possible social states is defined as:
$$
\mathbb{X}:=\mathbb{X}^1 \times \cdots \mathbb{X}^\rho.
$$
When the time argument is dropped, $x^r$ and $x$  represent a specific state for population $r$ and a specific social state, respectively. { We denote by $T\mathbb{X}$ the tangent space of $\mathbb{X}$, that is 
$
T\mathbb{X}=T\mathbb{X}^1 \times \cdots T\mathbb{X}^\rho$,
where $$T\mathbb{X}^r : =\left\{ z \in \mathbb{R}^{n^r}: \ \sum_{i=1}^{n^r}z_i=0\right\}, \ r=1,\dots,\rho.
$$
}

\begin{definition} {\bf (Memoryless payoff mechanism)}
Given a continuously differentiable map ${F:\mathbb{X} \rightarrow \mathbb{R}^n}$, the corresponding memoryless payoff mechanism, or population game, generates the payoff vector as follows:
\begin{equation} \label{eq:memorylesspayoffmech} p(t) = F\big(x(t) \big), \quad t\geq 0.
\end{equation} The state $\bar{x}$ is a Nash equilibrium if the following holds:
\begin{equation}\label{NE}
v^TF(\bar{x})\le \bar{x}^TF(\bar{x}), \quad  v\in \mathbb{X}.
\end{equation} 
We use $NE(F)$ to denote the set of Nash equilibria of $F$.
\end{definition}

The set $NE(F)$ is nonempty and closed   \cite[Proposition 1.2]{Weibull1995Evolutionary-ga}. In this paper we study the stability of  $NE(F)$ (in the sense of set stability \cite{LSW96}),  when the mean social state evolves according to the dynamics defined below.

\begin{definition} {\bf (EDM)}
We associate with each population, say $r$,  a Lipschitz continuous map ${\nu^r: \mathbb{X}^r\times \mathbb{R}^{n^r} \rightarrow \mathbb{R}^{n^r}}$ that specifies its evolutionary dynamics model (EDM) as follows:
\begin{equation}\label{EDMr}
\dot{x}^r(t)=\nu^r(x^r(t), p^r(t)), \quad t \geq 0
\end{equation}
where $p^r(t) \in \mathbb{R}^{n^r}$ is the $r$th conformal partition of $p(t)$ and the combined vector $x(t)\in \mathbb{X}$ evolves according to
\begin{equation}\label{EDM}
\dot{x}(t)=\nu(x(t), p(t)),  \ \mbox{where} \ \nu(x, p):=\begin{bmatrix} \nu^1(x^1, p^1) \\ \vdots \\ \nu^{\rho}(x^{\rho}, p^{\rho}) \end{bmatrix}.
\end{equation} 
{We assume that $\nu(x,p)$ belongs to the tangent cone~\cite{Sandholm2010Population-Game} to $\mathbb{X}$ at $x\in \mathbb{X}$ for any $p\in \mathbb{R}^n$, so that}
$x(t)$, the solution of~(\ref{EDM}) at time $t$, remains in $\mathbb{X}$ when $x(0)$ is in~$\mathbb{X}$.
\vspace{.05in}
\end{definition}

\begin{remark} {\bf (EDM and deterministic evolutionary dynamics)}
As  discussed in~\cite{Park2018Payoff-Dynamic-}, our definition of EDM is closely related to the concept of deterministic evolutionary dynamics described in~\cite{Sandholm2010Population-Game} specifically for memoryless payoff mechanisms of the form~(\ref{eq:memorylesspayoffmech}). In contrast to our approach in which we view the EDM~(\ref{EDM}) as a dynamical system whose input is $p$ and the output is $x$, the deterministic evolutionary dynamics is defined in~\cite[Section~4.4]{Sandholm2010Population-Game} as a set-valued map that assigns to each $F$ a set of mean social state trajectories \textemdash typically obtained as the solution of the initial value problem applied to the so-called mean dynamics derived in~\cite[Section~4.2.1]{Sandholm2010Population-Game}. The importance of the EDM concept for our approach is twofold: (i) it enables the analysis of the case where the payoff mechanism has internal dynamics specified by a PDM;
(ii) it allows us to break apart the stability analysis into separate steps for
establishing dissipativity properties for the EDM and the PDM, thus enabling a modular approach. 
\end{remark}

The method described in~\cite[Section~4.2.1]{Sandholm2010Population-Game} to obtain the mean dynamics for a given revision protocol can be used with minor changes to determine $\nu$ in~(\ref{EDM}). For instance, \cite[Example~4.3.5]{Sandholm2010Population-Game} can be adapted to obtain the Smith~EDM specified by $\nu^{\text{\tiny Smith}}$ given below for the $r$th population:
\begin{equation} \label{eq:SmithEdm}
\nu^{\text{\tiny Smith},r}_i(x^r,p^r) := \sum_{j=1}^{n^r} x^r_j \left[ p^r_i-p^r_j \right]_+ - x^r_i \left[ p_j^r-p_i^r\right]_+
\end{equation}  
where $\left[ s\right]_+:=\max\{s,0\}$.

The Smith EDM, which is of particular interest to Example~\ref{mixed} (to be discussed later on), is based on the revision protocol used in~\cite{Smith1984The-stability-o} for traffic assignment strategies. As one can infer from~(\ref{eq:SmithEdm}), the rate at which the agents in population $r$ switch from strategy $i$ to $j$ is proportional to $\left[ p_j^r-p_i^r\right]_+$. The following so-called impartial pairwise comparison (IPC) EDM is a generalization of Smith's:
\begin{equation}
\nu^{\text{\tiny IPC},r}_i(x^r,p^r) := \sum_{j=1}^{n^r} x^r_j \phi_i^r \left(p^r_i-p^r_j \right) - x^r_i \phi_j^r \left(p_j^r-p_i^r\right)
\end{equation} where, for each $j$ in $\{1,\ldots,n^r\}$, $\phi^r_j:\mathbb{R} \rightarrow \mathbb{R}_{\geq0}$ is a Lipschitz continuous function for which $\phi^r_j(\tilde{p}) >0$ if $\tilde{p} >0$ and $\phi^r_j(\tilde{p}) =0$ otherwise.

\begin{definition} {\bf (Nash Stationarity)}
We say that  an EDM described as in~(\ref{EDM}) 
has the ``Nash stationarity" property if the following equivalence holds:
\begin{equation}\label{NS}
\nu(x,p)=0 \quad \Leftrightarrow \quad \left[ \ v^Tp\le x^Tp, \quad v\in \mathbb{X}, \ p\in \mathbb{R}^n\ \right].
\end{equation}
\end{definition}
This condition ensures that the rest points for (\ref{EDM}) with $p=F(x)$ are Nash equilibria since, by (\ref{NS}), $x=\bar{x}$ satisfies
$\nu(x,p)=0$ and $p=F(x)$ if and only if (\ref{NE}) holds.
Note that (\ref{NS}) means that,  for each $r\in \{1,\dots,\rho\}$, the following holds:
\begin{equation}
\nu^r(x^r,p^r)=0 \quad \Leftrightarrow \quad \left[ \ \omega^Tp^r\le {x^r}^Tp^r, \quad  \omega \in \mathbb{X}^r \ \right].
\end{equation}
Thus, Nash stationarity of the multi-population EDM (\ref{EDM}) is equivalent to Nash stationarity for the EDM (\ref{EDMr}) of each population.

\begin{remark} {\bf (IPC EDM satisfies~(\ref{NS})) } \label{rem:NSEXamples} Examples of EDM satisfying Nash stationarity include those of the IPC type. Although in this article we frequently refer to the IPC class, there are other\footnote{See~\cite[Chapter~5]{Sandholm2010Population-Game} for more details.} important EDM classes that are Nash stationary, such as the so-called excess payoff target protocol (EPT) EDM. In fact, as we will mention in remarks throughout this article, certain sub-classes of the EPT EDM class  have some of the useful properties of the IPC EDM class.
\end{remark}

\section{$\delta$-Dissipativity and a Stability Theorem}\label{sec:delta-dissipative}

We now introduce the $\delta$-dissipativity property of an EDM, which can be viewed as a generalization of the notion of $\delta$-passivity proposed in~\cite{Fox2013Population-Game} and subsequently used in~\cite{Park2018Passivity-and-e,Park2019From-population,Park2018Payoff-Dynamic-} to  characterize the stability of Nash-like equilibria for the mean closed loop. 

\begin{definition}\label{def:deltaD} 
{\bf ($\delta$-Dissipativity w.s.r. $\Pi$)} An EDM specified by $\nu$ is $\delta$-dissipative with supply rate (w.s.r.) characterized by $\Pi=\Pi^T\in \mathbb{R}^{2n \times 2n}$ if there exist a continuously differentiable storage function ${S: \mathbb{X}\times \mathbb{R}^n \rightarrow \mathbb{R}_{\ge 0}}$ and a nonnegative-valued function $\sigma: \mathbb{X}\times \mathbb{R}^n \rightarrow \mathbb{R}_{\ge 0}$ that satisfy the following inequality for all $x$, $p$ and $u$ in $\mathbb{X}$, $\mathbb{R}^n$ and $\mathbb{R}^n$, respectively:
\begin{subequations}
\begin{eqnarray}\label{dissipation}
\!\!\!\!\!\! && \frac{\partial S(x,p)}{\partial x}\nu(x,p)+\frac{\partial S(x,p)}{\partial p}u \\ && \le -\sigma(x,p)+\begin{bmatrix} u \\ \nu(x,p) \end{bmatrix}^T \Pi \begin{bmatrix} u \\ \nu(x,p) \end{bmatrix} \nonumber,
\end{eqnarray}
where $S$ and $\sigma$ must also satisfy the equivalences below:
\begin{align}\label{negdef}
\sigma(x,p)=0 \quad &\Leftrightarrow \quad \nu(x,p)=0 \\
\label{informative}
S(x,p)=0 \quad &\Leftrightarrow \quad \nu(x,p)=0.
\end{align}
\end{subequations}
\end{definition}
\vspace{.1in}

\begin{remark}\label{rem:deltapassive-vs-dissipative} {\bf (When $\delta$-dissipativity implies $\delta$-passivity)}
If the $n \times n$ block-partitions of $\Pi$ are  $\Pi_{11}=0$, $\Pi_{22}=0$ and $\Pi_{12}=\Pi_{21}=\frac{1}{2}I$, then the inequality in~(\ref{dissipation}) with $u=\dot{p}$ implies  $\delta$-passivity, as defined in~\cite[Section 4.2]{Fox2013Population-Game}. 
The approach in~\cite{Park2019From-population,Park2018Payoff-Dynamic-} requires that the storage function is also ``informative" in an appropriately defined sense. Our additional requirements~(\ref{negdef})-(\ref{informative})  play a similar role. In fact, one can show for the aforementioned choice of $\Pi$ that $\delta$-dissipativity implies both that the EDM is $\delta$-passive and that it has an informative storage function, but the opposite does not hold because informativeness of the storage function does not imply~(\ref{negdef})-(\ref{informative}). The reason for this discrepancy is that our results establish global asymptotic stability, while~\cite{Park2019From-population,Park2018Payoff-Dynamic-} 
also allow  weaker notions of stability.
\end{remark}

The following remark describes important EDM classes that are $\delta$-dissipative.

\begin{remark} \label{rem:passivityIPC} {\bf ($\delta$-dissipativity of IPC EDM)}
As pointed out in Remark~\ref{rem:NSEXamples}, an EDM of the IPC type satisfies Nash stationarity. 
It is also $\delta$-dissipative w.s.r. $\Pi$ as chosen in Remark~\ref{rem:deltapassive-vs-dissipative} with the following storage function $S^{\text{\tiny IPC}}$ and associated $\sigma^{\text{\tiny IPC}}$ :
\begin{subequations}
\begin{align}
S^{\text{\tiny IPC}}(x,p)&:= \sum_{r=1}^\rho S^{\text{\tiny IPC},r}(x^r,p^r), \\ \sigma^{\text{\tiny IPC}}(x,p)&:= \sum_{r=1}^\rho \sigma^{\text{\tiny IPC},r}(x^r,p^r),
\end{align} \end{subequations} where for each population $r$:
\begin{subequations}
\label{eq:IPCEachPopStorageSigma}
\begin{align} 
S^{\text{\tiny IPC},r}(x^r,p^r) & :=\sum_{i=1}^{n^r} \sum_{j=1}^{n^r} x_i^r \int_0^{p_j^r-p_i^r} \phi^r_j(\tilde{p}) d\tilde{p} \\
\sigma^{\text{\tiny IPC},r}(x^r,p^r) &:= -  \sum_{i=1}^{n^r} \nu_i^{\text{\tiny IPC},r}(x^r,p^r) \sum_{j=1}^{n^r} \int_0^{p_j^r-p_i^r} \phi^r_j(\tilde{p}) d\tilde{p}.
\end{align}
\end{subequations}
As discussed in~\cite{Fox2013Population-Game},
the argument in \cite[Appendix~A.4]{Hofbauer2009Stable-games-an} can be used here to show that $\sigma^{\text{\tiny IPC}}$ is nonnegative and satisfies~(\ref{negdef})-(\ref{informative}). 
 We can also invoke an immediate analogy of this analysis, after appropriately modifying the arguments used to prove~\cite[Theorem~5.1]{Hofbauer2009Stable-games-an} and~\cite[Theorem~4.4]{Fox2013Population-Game}, to claim that the so-called separable EPT EDM class is also $\delta$-dissipative w.s.r. $\Pi$, for the choice of $\Pi$ in Remark~\ref{rem:deltapassive-vs-dissipative}.
\end{remark}

\vspace{-.15in}
\subsection{Stability of $NE(F)$ for memoryless payoff mechanisms}
We now investigate the stability of $NE(F)$ when the payoff $p$ accessed by the EDM is obtained from the memoryless map~(\ref{eq:memorylesspayoffmech}), leading to the following system:
\begin{equation} \label{eq:meandynamic}
\dot{x}(t)=\nu \Big( x(t), F \big( x(t) \big) \Big), \quad t \geq 0.
\end{equation}

{For the following theorem we assume a continuously differentiable extension of $F$ from $\mathbb{X}$ to $\mathbb{R}^n$ is available so that the Jacobian matrix $\frac{\partial F(x)}{\partial x}\in \mathbb{R}^{n\times n}$ is well defined for all $x\in \mathbb{X}$.}
\begin{theorem}\label{main} \rm
Suppose that the EDM (\ref{EDM}) is Nash stationary and $\delta$-dissipative w.s.r. $\Pi$, for a pre-selected~$\Pi$. If $F$ satisfies the following condition then $NE(F)$ is a globally asymptotically stable set of~(\ref{eq:meandynamic}):
{
\begin{equation}\label{incremental}
\zeta^T\begin{bmatrix} \frac{\partial F(x)}{\partial x} \\ I \end{bmatrix}^T \Pi  \ \begin{bmatrix} \frac{\partial F(x)}{\partial x} \\ I \ \end{bmatrix} \zeta \le 0 \quad \zeta\in T\mathbb{X}, \ x\in \mathbb{X}.
\end{equation}
}
\end{theorem}

\begin{proof} We show that $V(x)=S(x,F(x))$ serves as a Lyapunov function and guarantees global asymptotic stability of $NE(F)$.  $V(x)$ is nonnegative for all $x\in \mathbb{X}$ and, by (\ref{informative}) and (\ref{NS}), it vanishes only when $x\in NE(F)$. Next, note that
\begin{eqnarray*}
&&\!\!\!\!\!\!\!\!\!\!\!\! \frac{\partial V(x)}{\partial x}\nu(x,F(x))\\
&&\!\!\!\!\!\!\!\!\!\!\!\!\! = \! \left.\left\{\!\frac{\partial S(x,p)}{\partial x}\nu(x,p)\!+\!\frac{\partial S(x,p)}{\partial p}\frac{\partial F(x)}{\partial x}\nu(x,p)\! \right\}\right|_{p=F(x)}\\
&& \!\!\!\!\!\!\!\!\!\!\!\! \le -\sigma(x,F(x))+ \\
&& 
\left.\left\{\begin{bmatrix} u \\ \nu(x,p) \end{bmatrix}^T \Pi \begin{bmatrix} u \\ \nu(x,p) \end{bmatrix} \right\}\right|_{u=\frac{\partial F(x)}{\partial x}\nu(x,p), \, p=F(x)} \\
 &&\!\!\!\!\!\!\!\!\!\!\!\! =
-\sigma(x,F(x)) + \\
 &&\!\!\!\!
\nu(x,F(x))^T  \begin{bmatrix} \frac{\partial F(x)}{\partial x} \\ I \end{bmatrix}^T \Pi \begin{bmatrix} \frac{\partial F(x)}{\partial x} \\ I \end{bmatrix} \nu(x,F(x)),
\end{eqnarray*}
where we have used (\ref{dissipation}). Then, { since $\nu(x,F(x)) \in T\mathbb{X}$,} the inequality (\ref{incremental}) implies
$$
\frac{\partial V(x)}{\partial x}\nu(x,F(x)) \le -\sigma(x,F(x)),
$$
where the right-hand side is nonpositive and, by (\ref{negdef}), vanishes when $x\in NE(F)$. Thus $V(x)=S(x,F(x))$ indeed serves as a Lyapunov function and guarantees global asymptotic stability of $NE(F)$.  
\end{proof}

In the following remark  we use Theorem~\ref{main} to recover key stability results from~\cite{Hofbauer2009Stable-games-an}. To do so, we first note that,  when $\Pi_{11}=0$, the inequality in~(\ref{incremental}) is equivalent to the following incremental quadratic constraint:
\begin{equation}\label{incremental2}
\begin{bmatrix} F(x)-F(y) \\ x-y \end{bmatrix}^T \Pi  \ \begin{bmatrix} F(x)-F(y) \\ x-y  \end{bmatrix} \le 0, \quad  x,y\in \mathbb{X}.
\end{equation}

\begin{remark}\label{rem:Theorem1-rec1} {\bf (Recovering key results from~\cite{Hofbauer2009Stable-games-an})} 
If the $n \times n$ block-partitions of $\Pi$ are  $\Pi_{11}=0$, $\Pi_{22}=0$ and $\Pi_{12}=\Pi_{21}=\frac{1}{2}I$ as previously considered in Remark~\ref{rem:deltapassive-vs-dissipative}, 
then (\ref{incremental2}) becomes the following {\it contraction} inequality:
\begin{equation}\label{con}
(F(x)-F(y))^T(x-y) \le 0, \quad  x,y\in \mathbb{X}.
\end{equation}
Hence, Remark~\ref{rem:passivityIPC} allows us to invoke Theorem~\ref{main} to recover the portions of~\cite[Theorems~5.1~and~7.1]{Hofbauer2009Stable-games-an} that guarantee that $NE(F)$ is a globally asymptotically stable set for~(\ref{eq:meandynamic}) when the EDM is a separable EPT or IPC, provided that $F$ satisfies~(\ref{con}). In this context, it is also relevant to mention that \cite[Corollary~2]{Park2018Payoff-Dynamic-} extends~\cite[Theorem~5.1]{Hofbauer2009Stable-games-an} to the more general class of integrable EPT EDM even when $NE(F)$ is not a singleton.
 \end{remark}

As is noted on the comparison in~\cite[Section~2.4]{Hofbauer2009Stable-games-an}, negative definiteness conditions, more precisely strict diagonal concavity, have been proposed in~\cite{Rosen1965Existence-and-u} to establish the uniqueness of Nash equilibria for certain normal form games. The comparison, however, explains that not only is the context in~\cite{Rosen1965Existence-and-u} rather distinct from what we consider here, but, even in strictly mathematical terms, contractivity is analogous to diagonal concavity only in the very particular case in which $F$ has no own-population interactions. 

\begin{remark}\label{rem:Theorem1-rec2} {\bf (Characterizing $\delta$-passivity surplus)}
Another special case of Theorem \ref{main} is when $\Pi_{11}=0$, $\Pi_{12}=\Pi_{21}=\frac{1}{2}I$, $\Pi_{22}=-\eta I$, $\eta>0$, in which case the EDM is said to have a  ``surplus" of $\delta$-passivity \cite{Park2018Payoff-Dynamic-,Park2019From-population}. With this choice of $\Pi$, the condition in~(\ref{incremental2}) reduces to the following inequality indicating that a commensurate ``deficit" of contraction is  allowed in the payoff model:
$$
(F(x)-F(y))^T(x-y) \le \eta \|x-y\|^2, \quad  x,y\in \mathbb{X}.
 $$As shown in~\cite[Corollary~IV.3]{Park2018Passivity-and-e,Park2018Passivity-and-E2}, the EPT and IPC EDM classes do not have $\delta$-passivity surplus when $n^r \geq 3$. However, there are instances of the so-called perturbed best response\footnote{See~\cite{Hofbauer2007Evolution-in-ga} for a in-depth analysis of the PBR revision protocol.} (PBR) EDM class that have $\delta$-passivity surplus. The analysis of $\delta$-passivity for the PBR EDM class was first put forth in~\cite{Park2018Passivity-and-e,Park2018Passivity-and-E2}, and later extended in~\cite{Park2018Payoff-Dynamic-}. Because this class is not Nash stationary, the analysis of $\delta$-dissipativity of the PBR EDM  is beyond the scope of this article. 
 \end{remark}

\section{Weighted Contractive Games}\label{weighted}

In Remark~\ref{rem:Theorem1-rec1}, we explained how Theorem~\ref{main} recovers known results for contractive games as special cases. To demonstrate the broader applicability of Theorem \ref{main}, we now consider a multi-population game in which the EDM (\ref{EDMr}) is $\delta$-passive for each population, and show that stability of Nash equilibria can be ascertained with a relaxed form of the contraction property (\ref{con}) for the payoff model. We start by stating the following lemma, which assumes the EDM for each population  $r$ is $\delta$-dissipative w.s.r. $\Pi^r$, and constructs a composite $\Pi$ with flexible weights. We state the lemma without proof, as it follows imediately from Definition~\ref{def:deltaD}.

\begin{lemma}\label{multi-passive}\rm 
Suppose that  each population $r\in \{1,\cdots,\rho\}$  of an  EDM (\ref{EDMr}) possesses a storage function $S^r: \mathbb{X}^r\times \mathbb{R}^{n^r} \rightarrow \mathbb{R}_{\ge 0}$ and ${\sigma^r: \mathbb{X}^r\times \mathbb{R}^{n^r} \rightarrow \mathbb{R}_{\ge 0}}$ 
 satisfying the  conditions:
\begin{subequations}
\label{eq:conditionsLemmaMultipop}
\begin{multline}\label{dissipat}
\frac{\partial S^r(x^r,p^r)}{\partial x^r}\nu^r(x^r,p^r)+\frac{\partial S^r(x^r,p^r)}{\partial p^r}u^r 
 \le -\sigma^r(x^r,p^r)
  \\
+\begin{bmatrix} u^r \\ \nu^r(x^r,p^r)\end{bmatrix}^T \begin{bmatrix}\Pi_{11}^r & \Pi_{12}^r \\ \Pi_{21}^r & \Pi_{22}^r \end{bmatrix}\begin{bmatrix} u^r \\ \nu^r(x^r,p^r)\end{bmatrix}, 
\end{multline}
\begin{align}\label{informative-sub}
S^r(x^r,p^r)=0 &\quad \Leftrightarrow \quad \nu^r(x^r,p^r)=0, \\
\sigma^r(x^r,p^r)=0 &\quad \Leftrightarrow \quad \nu^r(x^r,p^r)=0
\end{align} for all $x^r\in \mathbb{X}^r$, $p^r\in \mathbb{R}^{n^r}$, and $u^r\in \mathbb{R}^{n^r}$.
\end{subequations}
Then, for any given choice of positive weights $w^1, \ldots, w^\rho$, the composite storage function 
$$
S(x,p)=\sum_{r=1}^{\rho}w^rS^r(x^r,p^r)
$$
for the multi-population EDM (\ref{EDM}) satisfies~(\ref{dissipation}) with
\begin{equation}\label{general}
\Pi\!=\!\!\begin{bmatrix}w^1\Pi_{11}^1 & & &w^1\Pi_{12}^1& & \\ &\ddots & & & \ddots &  \\   & & w^\rho\Pi_{11}^\rho & & &w^\rho\Pi_{12}^\rho \\ 
w^1\Pi_{21}^1 & & & w^1\Pi_{22}^1& & \\ &\ddots & & & \ddots &  \\   & & w^\rho\Pi_{21}^\rho & & &w^\rho\Pi_{22}^\rho 
\end{bmatrix}.
 \end{equation}
  In particular, if the EDM for each population is $\delta$-passive ($\Pi_{11}^r=\Pi_{22}^r=0$, $\Pi_{12}^r=\Pi_{21}^r=\frac{1}{2}I$), then the multi-population EDM (\ref{EDM}) satisfies~(\ref{dissipation}) with
\begin{subequations}
\label{eq:multipop-passive}
\begin{align}
\Pi_{11}=\Pi_{22}=0, \quad \Pi_{12}=\Pi_{21}=\frac{1}{2}W\\
\label{multiplier}
 W:=\begin{bmatrix}w^1I_{n^1} & & \\ & \ddots & \\ & & w^\rho I_{n^\rho} \end{bmatrix}.
\end{align}
\end{subequations}
\end{lemma}
\vspace{.1in}
\begin{remark}\label{rem:IPCIsEachPopDiss} {\bf (Lemma~\ref{multi-passive} applies to IPC EDM)} Using the same analysis as Remark~\ref{rem:passivityIPC}, we infer that $S^{\text{\tiny IPC},r}$ and $\sigma^{\text{\tiny IPC},r}$ for the IPC EDM class, as given in~(\ref{eq:IPCEachPopStorageSigma}), satisfy (\ref{eq:conditionsLemmaMultipop}) with $\Pi_{11}^r=\Pi_{22}^r=0$, $\Pi_{12}^r=\Pi_{21}^r=\frac{1}{2}I_{n^r}$. Therefore the conclusion of Lemma~\ref{multi-passive} holds with (\ref{eq:multipop-passive}) for this class.
\end{remark}

The following corollary to Lemma~\ref{multi-passive} shows that we can leverage the flexible weights $w^1, \ldots, w^\rho$ to relax condition (\ref{incremental}) of
Theorem~\ref{main}. In particular, when the EDM for each population is $\delta$-passive, we can  establish global asymptotic stability for the Nash equilibrium set when the payoff is not necessarily contractive, but becomes so upon an appropriate choice of weights.

\begin{corollary}{\bf (Weighted contraction)} \label{cor} \rm Under the hypotheses of Lemma \ref{multi-passive}, $NE(F)$ is a globally asymptotically stable equilibrium set of~(\ref{eq:meandynamic}) if  there exist positive weights $w^1, \ldots, w^\rho$ with which (\ref{general}) satisfies (\ref{incremental}).
In particular, when the EDM for each population is $\delta$-dissipative with $\Pi_{11}^r=\Pi_{22}^r=0$, $\Pi_{12}^r=\Pi_{21}^r=\frac{1}{2}I_{n^r}$, global asymptotic stability follows if
the following holds:
{
\begin{equation} \label{weighted con}
\zeta^T\left(W\frac{\partial F(x)}{\partial x}+\frac{\partial F(x)}{\partial x}^TW\right)\zeta \le 0, \quad \zeta \in T\mathbb{X}, \ x \in \mathbb{X}
\end{equation}
}
for some $W>0$ of the form (\ref{multiplier}) or, equivalently, 
\begin{equation}
(WF(x)-WF(y))^T(x-y) \le 0, \quad x,y \in \mathbb{X}.
 \end{equation}
 \end{corollary}
  
\subsection{Example: Congestion Game with Mixed Autonomy}
\label{mixed}
As an illustration of weighted contractive games consider a road network, described as a directed graph where each link is a road segment connecting two distinct nodes. Suppose there are $\gamma \ge1$ origin-destination (OD) node pairs and, for each pair $r=1,\dots,\gamma$, there exist $n^r\ge 1$ routes
that traverse no link twice and connect  the origin to the destination.
 
Following \cite{LazCooPed17,MehHor19} we consider two types of vehicles - autonomous and regular - for each OD pair, resulting in $\rho=2\gamma$ populations.
We denote by 
$x^r$,  $r=1,\cdots,\gamma$, the flow vector of autonomous vehicles with OD pair $r$, and by $x^{\gamma+r}$ the flow vector of regular vehicles with OD pair $r=1,\cdots,\gamma$. 
 We further define 
\begin{equation}
\label{qe:xstructureexamplemixed}
\mathbf{x}^{\rm aut}:=\begin{bmatrix} x^1 \\ \vdots \\ x^{\gamma} \end{bmatrix}, \quad \mathbf{x}^{\rm reg}:=\begin{bmatrix} x^{\gamma+1} \\ \vdots \\ x^{2\gamma}\end{bmatrix},  \quad x=\begin{bmatrix} \mathbf{x}^{\rm aut}  \\ \mathbf{x}^{\rm reg}\end{bmatrix}.
\end{equation}
Since the constituent vectors $x^r$ and  $x^{\gamma+r}$ each has $n^r$ entries,  $\mathbf{x}^{\rm aut}$ and $\mathbf{x}^{\rm reg}$ have $N:=n^1+\cdots+n^\gamma$ entries, one for each route.

Next we 
let $L$ denote the number of links in the graph and define the $N\times L$ routing matrix 
$$
R_{i\ell}=\left\{  \begin{array}{ll} 1 & \mbox{if route $i$ traverses link $\ell$} \\ 0 & \mbox{otherwise} \end{array}\right.
$$
and note that $z^{\rm aut}:=R^T\mathbf{x}^{\rm aut}$ and $z^{\rm reg}:=R^T\mathbf{x}^{\rm reg}$ are vectors of link flows, the former due to autonomous vehicles and the latter due to regular vehicles. 
As in \cite{MehHor19} we 
assume the delay on link $\ell$ is an increasing function 
$\Phi_\ell$ of 
$$
z_\ell:=\mu z^{\rm aut}_\ell+z^{\rm reg}_\ell,
$$
where the factor $\mu\in (0,1)$ accounts for the shorter headway maintained by autonomous vehicles. 

Indeed, a shorter headway increases the capacity of the link,
therefore the delay incurred on  link $\ell$ is better represented as a function of $z_\ell$ defined above, instead of the unweighted sum $z^{\rm aut}_\ell+z^{\rm reg}_\ell$ that does not discriminate between autonomous and regular vehicles.
We may then define a cost function for vehicles using route $i$  (autonomous or regular) as the sum of the delays incurred on each link  traversed by route $i$, and assign the negative of the cost function as the payoff:
\begin{equation}\label{payoff}
F_{i}(x)=F_{i+N}(x)=-\sum_{\ell=1}^L R_{i\ell}\Phi_\ell(z_\ell), \quad i=1,\dots,N.
\end{equation}
Then the full payoff vector is
\begin{equation}\label{payoff-vec}
p=F(x)=-\begin{bmatrix}R \\ R \end{bmatrix}\begin{bmatrix} \Phi_1(z_1) \\ \vdots \\ \Phi_L(z_L) \end{bmatrix}
\end{equation}
where
\begin{equation}\label{zdef}
 z=\mu R^T \mathbf{x}^{\rm aut}+R^T \mathbf{x}^{\rm reg}=\begin{bmatrix}R \\ R \end{bmatrix}^T\begin{bmatrix}\mu I & 0 \\ 0 & I \end{bmatrix} x.
\end{equation}
It follows that
\begin{equation}\label{Jacobian}
\frac{\partial F(x)}{\partial x}\!=\!-\begin{bmatrix}R \\ R \end{bmatrix} \!\! \begin{bmatrix} \Phi_1'(z_1) & &  \\  & \ddots &  \\ & & \Phi'_{L}(z_L)\end{bmatrix}\!\!  \begin{bmatrix}R \\ R \end{bmatrix}^T\!\!  \begin{bmatrix}\mu I & 0 \\ 0 & I \end{bmatrix}\!
\end{equation}
where $\Phi_{\ell}'$ denotes the derivative of $\Phi_{\ell}$, $\ell=1,\dots,L$.
Since each $\Phi_\ell$ is an increasing function,  the diagonal entries $\Phi_1', \dots, \Phi'_L$ above are nonnegative,
and
$$
\begin{bmatrix}\mu I & 0 \\ 0 & I \end{bmatrix}\frac{\partial F(x)}{\partial x}
$$
is symmetric and negative semidefinite. Therefore, (\ref{weighted con}) holds with 
\begin{equation}\label{W}
W=\begin{bmatrix}\mu I & 0 \\ 0 & I \end{bmatrix},
\end{equation}
and we conclude from Corollary~\ref{cor} that, if the EDM for each $r=1,\cdots, 2\gamma$ is $\delta$-dissipative with $\Pi_{11}^r=\Pi_{22}^r=0$, $\Pi_{12}^r=\Pi_{21}^r=\frac{1}{2}I_{n^r}$, then the set of Nash equilibria is globally asymptotically stable.
Hence, in light of the Remark~\ref{rem:IPCIsEachPopDiss} and Corollary~\ref{cor}, we can state without proof the following corollary characterizing the stability of $NE(F)$ for Example~\ref{mixed} when the EDM is of the IPC class.
\vspace{.1in}

\begin{corollary} \label{cor:stabilityExample1} If $F$ is as given in (\ref{payoff}), and the EDM is of the IPC class, such as the Smith EDM, then $NE(F)$ is a globally asymptotically stable equilibrium set of~(\ref{eq:meandynamic}).
\end{corollary}

\section{Constant Matrix Parameterizations of the Payoff Jacobian}
\label{sec:Search}
In applications it may be difficult to verify that condition (\ref{incremental}) of Theorem \ref{main} holds for all $x\in \mathbb{X}$. To overcome this difficulty we propose bounding the Jacobian matrix 
$$
J(x):=\frac{\partial F(x)}{\partial x}
$$
within a set parameterized by constant matrices. Such parameterizations include the convex hull:
\begin{eqnarray}
&&\!\!\!\!\!\!\!\!\!\!\!\!\!\!\! {\rm conv}\{ A_1,\dots,A_k \}
:=\left\{ \lambda_1A_1+\cdots +\lambda_kA_k : \right.\\ \nonumber
&&\quad \quad \left. \lambda_i\ge 0, i=1,\dots,k, \ \lambda_1+\cdots+\lambda_k=1 \right\},
\end{eqnarray}
and the conic hull:
\begin{eqnarray}
&&\!\!\!\!\!\!\!\!\!\!\!\!\!\!\!\!\!\!\!
{\rm cone}\{ B_1,\dots, B_s \}
:=\left\{ \theta_1B_1+\cdots +\theta_sB_s  : \right. \\ \nonumber
&&\qquad \qquad \qquad \qquad \quad \left. 
 \theta_i\ge 0, i=1,\dots,s \right\}.
\end{eqnarray}
If $J(x)$ lies in one of these sets or their sum  for all $x\in \mathbb{X}$, then we can ascertain condition (\ref{incremental}) by checking matrix inequalities involving only the constant matrices $A_1,\dots, A_k$, $B_1,\dots, B_s$.
\begin{proposition}\label{prop}\rm
Let $\Pi_{11}=0$ { and let $P\in \mathbb{R}^{n\times n}$ be the orthogonal projection matrix onto $T\mathbb{X}$}. Then either of the following conditions guarantees (\ref{incremental}):

\noindent
i) $J(x)\in {\rm conv}\{ A_1,\dots,A_k \}$  for all $x\in \mathbb{X}$, and
\begin{equation}\label{conv}
P\begin{bmatrix} A_i \\ I \end{bmatrix}^T \Pi  \ \begin{bmatrix}A_i \\ I \ \end{bmatrix}P \le 0 \quad i=1,\cdots,k.
\end{equation}

\noindent
ii) $J(x)\in {\rm cone}\{ B_1,\dots, B_s \}$ for all $x\in \mathbb{X}$, $P\Pi_{22}P\le 0$, and
\begin{equation}\label{cone}
P(\Pi_{12}^TB_i+B_i^T\Pi_{12})P\le 0  \quad i=1,\cdots,s.
\end{equation}

\noindent
iii) $J(x)\in {\rm conv}\{ A_1,\dots,A_k \}+ {\rm cone}\{ B_1,\dots, B_s \}$ for all $x\in \mathbb{X}$, and (\ref{conv}) and (\ref{cone}) hold. \hfill $\Box$
\end{proposition}

As an illustration, the Jacobian (\ref{Jacobian}) in Example \ref{mixed} can be rewritten as
$$
J(x)=\sum_{\ell=1}^L \Phi'_\ell(z_\ell) B_\ell,
$$
where
$$
B_\ell=\begin{bmatrix}R_\ell \\ R_\ell \end{bmatrix} \begin{bmatrix}\mu R_\ell^T & R_\ell^T \end{bmatrix}
$$
and $R_\ell$ denotes the $\ell$th column of the routing matrix $R$. Since $\Phi'_\ell(z_\ell)\ge 0$ for each $\ell$, we conclude that $J(x)\in {\rm cone}\{ B_1,\dots, B_L \}$, and (\ref{cone}) holds with $\Pi_{12}=W$ specified in (\ref{W}).

Recall that, in Section \ref{weighted}, we considered multi-population games where the EDM for each population satisfies a $\delta$-dissipativity property, leading to  the form  of $\Pi$ in (\ref{general}) with flexible weights $w^r>0$, $r=1,\cdots,\rho$. 
Since $\Pi$ depends linearly on these weights, conditions (\ref{conv})-(\ref{cone}) become linear matrix inequalities (LMIs) with decision variables $w^r>0$, $r=1,\cdots,\rho$.  Thus, we can search for weights satisfying (\ref{incremental}) numerically with convex programming software, such as CVX \cite{cvx}.

One may also encounter situations where
 the Jacobian belongs to a convex set of the form:
\begin{eqnarray}\nonumber
&& \!\!\!\!\!\!\!\!\!\!\!\!\!\! {\rm box}\{ G_0, G_1,\dots, G_d \}
:=\left\{ G_0+ \gamma_1G_1+\cdots +\gamma_dG_d  : \right. \\ \label{box}
&& \qquad \qquad  \qquad \qquad \left. \gamma_i \in [0,1], i=1,\dots,d \right\}.
\end{eqnarray}
Although we can apply Proposition \ref{prop}(i) to the matrices that form the vertices of this set, this application involves $k=2^d$ vertices and may become intractable for large $d$. We next propose an alternative test to check (\ref{incremental}) that involves only $d+1$ matrices, $G_0, G_1,\dots, G_d$:
\begin{proposition}\label{prop2}\rm
Let $\Pi_{11}=0$ and {and define $P\in \mathbb{R}^{n\times n}$ as in Proposition \ref{prop}.} Suppose $J(x)\in {\rm box}\{ G_0, G_1,\dots, G_d \}$ where 
\begin{equation}\label{BC}
G_i=C_iD_i^T, \quad i=1,\dots,d,
\end{equation}
with $C_i,D_i\in \mathbb{R}^{n\times \varrho_i}$, $\varrho_i$ denoting the rank of $G_i$.
 Then (\ref{incremental}) holds if there exist constants $\omega_1,\cdots, \omega_d>0$ s.t.
\begin{equation}\label{Scheck}
\begin{bmatrix} 
P(\Pi_{12}^TG_0+G_0^T\Pi_{12}+\Pi_{22})P & P(\Pi_{12}^TC+D\Omega) \\ (C^T\Pi_{12}+\Omega D^T)P & -2\Omega
\end{bmatrix}
\le 0,
\end{equation}
where $C:=\begin{bmatrix} C_1 \cdots C_d \end{bmatrix}$, $D:=\begin{bmatrix} D_1 \cdots D_d \end{bmatrix}$, and
\begin{equation}\label{Omega}
 \Omega:=\begin{bmatrix} \omega_1 I_{\varrho_1} & & \\ & \ddots & \\ & & \omega_d I_{\varrho_d} \end{bmatrix}.
\end{equation}
\end{proposition}
Note that  (\ref{Scheck}) is linear  in $\Omega$ and $\Pi$. Thus, the search for nonnegative constants $\omega_1,\cdots, \omega_d$ in (\ref{Omega}) to satisfy (\ref{Scheck}) can be performed numerically with LMI solvers. This search can also be combined with a simultaneous search for the weights  $w^1,\cdots,w^\rho$ when $\Pi$ has the form (\ref{general}) arising in multi-population games.

\vspace{-.1in}
\subsection{Example: Bypassing Near a Road Split}
Reference \cite{mehr2018game} developed a game theoretic model of lane changing behavior as drivers approach traffic diverges. In this model each vehicle selects lanes according to an appropriately defined payoff that accounts for crossing effects due to bypassing vehicles and the additional distance traveled by such vehicles. Here we will  use the numerical methods proposed above to show that this payoff model is contractive.
 Consider two populations of vehicles approaching a split, where population~$1$ is headed towards the first branch and population~$2$  towards the second. Following \cite{mehr2018game} we consider two strategies for each population: {\it steadfast} behavior where the vehicle stays on the lane destined to the branch, and {\it bypassing} behavior where the vehicle uses the other lane and merges with the correct lane before the split (see Figure \ref{Fig:split}). 

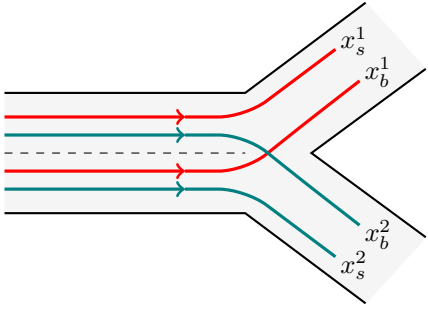
\begin{figure}[h] 
\begin{center}
\begin{tikzpicture}[scale=.08]
\coordinate (A1) at (0,10);
\coordinate (A2) at (40,10);
\coordinate (A3) at (60,25);
\coordinate (A4) at (70,14);
\coordinate (A5) at (51,0);
\coordinate (A6) at (0,-10);
\coordinate (A7) at (40,-10);
\coordinate (A8) at (60,-25);
\coordinate (A9) at (70,-14);
\draw[fill=black!4,black!4] (A1) -- (A2) -- (A3) -- (A4) -- (A5) -- (A9) -- (A8) -- (A7) -- (A6);

\draw[thin,dashed] (0,0) -- (40,0);

\draw[very thick,red,->] (0,6) -- (30,6);
\draw[very thick,red] (30,6)  [rounded corners=10pt] -- (40,6) -- (55,17.25) node[above right,black,yshift=-6,xshift=-2]{$x_s^1$};

\draw[very thick,red,->] (0,-3) -- (30,-3);
\draw[very thick,red] (30,-3)  [rounded corners=10pt] -- (40,-3) -- (50,5) -- (59,12) node[above right,black,yshift=-6,xshift=-2]{$x_b^1$};

\draw[very thick,teal,->] (0,-6) -- (30,-6);
\draw[very thick,teal] (30,-6)  [rounded corners=10pt] -- (40,-6) -- (55,-17.25) node[below right,black,yshift=6,xshift=-2]{$x_s^2$};

\draw[very thick,teal,->] (0,3) -- (30,3);
\draw[very thick,teal] (30,3)  [rounded corners=10pt] -- (40,3) -- (50,-5) -- (59,-12) node[below right,black,yshift=6,xshift=-2]{$x_b^2$};

\draw[thick] (A1) -- (A2) -- (A3);
\draw[thick] (A4) -- (A5) -- (A9);
\draw[thick] (A8) -- (A7) -- (A6);
\end{tikzpicture}
\end{center}
\caption{Lane changing behavior near a road diverge: {\it steadfast} drivers stay on the lane leading to their destination, while {\it bypassing} drivers use the other lane before merging with the correct one.}
\label{Fig:split}
\end{figure}

Let $x=[x^1_s \ x^1_b \ x^2_s \ x^2_b]^T$ where
$x_s^i$ and $x_b^i$ are the flows of steadfast and bypassing vehicles headed to branch $i=1,2$, and assume they are normalized by the total flow. That is, $x_s^i+x_b^i=m^i$, $i=1,2$, and $m^1+m^2=1$. Reference \cite{mehr2018game} proposes the payoff
\begin{equation}\label{Fex2}
F(x)=-\begin{bmatrix} c_1^t(x^1_s+x^2_b)+c_1^cx_b^1(x^1_s+x^2_b) \\ c_2^t(x^2_s+\vartheta_1x^1_b)+c_2^cx_b^2(x^2_s+x^1_b) \\  c_2^t(x^2_s+x^1_b)+c_2^cx_b^2(x^2_s+x^1_b) \\ c_1^t(x^1_s+\vartheta_2x^2_b)+c_1^cx_b^1(x^1_s+x^2_b)   \end{bmatrix}
\end{equation}
where $c_i^t>0$, $c_i^s>0$, and $\vartheta_i > 1$. Here $x^1_s+x^2_b$ is the fraction of the total flow using the lane destined to branch $1$ and $c_1^t(x^1_s+x^2_b)$ is the cost of traversing this lane. For bypassing vehicles the traversal cost is modified as  $c_1^t(x^1_s+\vartheta_2x^2_b)$, where $\vartheta_2>1$ accounts for the additional distance traveled.
The other term, $c_1^cx_b^1(x^1_s+x^2_b)$, appearing in the first and fourth entries of $F$ is the cost of crossing effects due to bypassing vehicles $x^1_b$ merging to the lane for branch~$1$.

Recall $x_s^i$ and $x_b^i$, $i=1,2$, are normalized by the total flow, and define
$
\gamma_1=x_s^1, \ \gamma_2=x_b^1, \ \gamma_3=x_s^2,  \ \gamma_4=x_b^2
$
so that $\gamma_i\in [0,1]$, $i=1,2,3,4$. Let $e_1,e_2,e_3, e_4$ denote the unit vectors in $\mathbb{R}^4$. Then, the Jacobian of (\ref{Fex2}) can be written as in (\ref{box})-(\ref{BC}) with $d=4$,  $C_1=-(e_1+e_4)$, $C_2=-\begin{bmatrix} e_1+e_4 & e_2+e_3 \end{bmatrix}$, $C_3=-(e_2+e_3)$, $C_4=-\begin{bmatrix} e_1+e_4 & e_2+e_3 \end{bmatrix}$, $D_1=c_1^ce_2$, 
$D_2=\begin{bmatrix} c_1^c(e_1+e_4) & c_2^ce_4 \end{bmatrix}$,  $D_3=c_2^ce_4$, $D_4=\begin{bmatrix} c_1^ce_2 & c_2^c(e_2+e_3) \end{bmatrix}$, 
$$
G_0=-
\begin{bmatrix}
 c_1^t & 0 & 0 & c_1^t
 \\ 0 & c_2^t\vartheta_1 & c_2^t & 0 
 \\ 0 & c_2^t & c_2^t & 0 
 \\ c_1^t & 0 & 0 & c_1^t\vartheta_2 
 \end{bmatrix}.
$$
For the values $c_1^t=c_2^t=c_1^c=c_2^c=1$, $\vartheta_1=\vartheta_2=2.7$ obtained from data in \cite{mehr2018game}, we ascertained using CVX \cite{cvx} that the LMI (\ref{Scheck}) is feasible for $\Pi_{12}=\frac{1}{2}I$, $\Pi_{22}=0$.
Therefore, the payoff (\ref{Fex2}) is contractive. Stability of Nash equilibria can then be established from the $\delta$-passivity of the EDM as in Theorem \ref{main} and the ensuing discussion, complementing the static analysis in \cite{mehr2018game} for the existence and uniqueness of a Nash equilibrium.

 \section{Dynamical Models for Payoff}
 \label{sec:DynamicPayoff}
We next consider the situation where, instead of the static model $p=F(x)$, the payoff evolves according to a dynamical model of the form
\begin{subequations}\label{PDM}
\begin{eqnarray}\label{PDM1}
\dot{q}(t)&=&f\big(q(t),x(t)\big)\\
p(t)&=&h\big(q(t),x(t)\big),\label{PDM2}
\end{eqnarray}
\end{subequations}
where $q(t)\in \mathbb{R}^o$, while ${f: \mathbb{R}^o \times \mathbb{X} \rightarrow  \mathbb{R}^o}$ and ${h: \mathbb{R}^o \times \mathbb{X} \rightarrow  \mathbb{R}^n}$ are Lipschitz continuous maps.

\begin{definition} {\bf (PDM)} Following \cite[Definition~4]{Park2018Payoff-Dynamic-} we refer to (\ref{PDM}) as the Payoff Dynamics Model (PDM), and assume that it recovers the static model $p=F(x)$ in steady-state, that is,
\begin{equation}\label{ssIO}
f(q,x)=0 \quad \Rightarrow \quad h(q,x)=F(x).
\end{equation}
\end{definition}

We now  generalize Theorem~\ref{main} to  the following closed loop model for the EDM~(\ref{EDM}) in feedback with~(\ref{PDM}):
\begin{subequations}
\label{eq:MeanClosedLoopModel}
\begin{align}
\dot{x}(t) & = \nu \Big ( x(t), h\big(q(t),x(t)\big) \Big), \\
\dot{q}(t)&=f\big(q(t),x(t)\big), \quad t \geq 0.
\end{align}
\end{subequations}

\begin{theorem}\rm \label{PDMthm}
Suppose that a Nash stationary EDM (\ref{EDM}), a PDM (\ref{PDM}) satisfying (\ref{ssIO}), and $\Pi=\Pi^T$ in $\mathbb{R}^{n \times n}$ are given. Under these conditions, the $x$ components of the rest points of~(\ref{eq:MeanClosedLoopModel}) constitute the set of Nash equilibria $NE(F)$. Moreover, the set of rest points is  globally asymptotically stable  for~(\ref{eq:MeanClosedLoopModel}) if the EDM is $\delta$-dissipative w.s.r. $\Pi$ and the PDM has a storage function $Q: \mathbb{R}^o \times \mathbb{X} \rightarrow  \mathbb{R}_{\ge 0}$ and a nonnegative map $\varsigma: \mathbb{R}^o \times \mathbb{X} \rightarrow  \mathbb{R}_{\ge 0}$ such that the following holds for all $q\in \mathbb{R}^o$, $x\in \mathbb{X}$, and {$\zeta \in T\mathbb{X}$}:
\begin{subequations}
\label{eq:DeltaAntidissipat}
\begin{align}\label{PDMdissip}
 \frac{\partial Q(q,x)}{\partial q}f(q,x)+\frac{\partial Q(q,x)}{\partial x}\zeta & \le -\varsigma(q,x) -\psi^T\Pi \psi \\
\label{vanish1} 
Q(q,x) = 0 \ & \Leftrightarrow \ f(q,x)=0 \\
\label{vanish2}
\varsigma(q,x)  =0 \ & \Leftrightarrow \ f(q,x)=0,
\end{align}
\end{subequations}
where $\psi$ is defined as:
\begin{equation}\label{psi-def}
\psi:=\begin{bmatrix} \frac{\partial h(q,x)}{\partial q}f(q,x)+\frac{\partial h(q,x)}{\partial x}\zeta \\ \zeta \end{bmatrix}.
\end{equation}
\end{theorem}

Before we give a proof, in the following remark we compare~(\ref{eq:DeltaAntidissipat}) with the related concept of $\delta$-antipassivity, as defined in~\cite[Definition~12]{Park2019From-population}, and also discussed in~\cite[Section~VI.B)]{Park2018Payoff-Dynamic-}.  $\delta$-antipassivity was originally defined in~\cite{Fox2013Population-Game} as the ``antisymmetric" of $\delta$-passivity without the additional condition~\cite[(54a)]{Park2019From-population}, which is needed to ascertain  stability in~\cite{Park2019From-population,Park2018Payoff-Dynamic-}.

\begin{remark} {\bf (When (\ref{eq:DeltaAntidissipat}) implies $\delta$-antipassivity)} 
Suppose $\Pi$ is chosen as in Remark~\ref{rem:deltapassive-vs-dissipative}, where the $n \times n$ block-partitions of $\Pi$ are  $\Pi_{11}=0$, $\Pi_{22}=0$ and $\Pi_{12}=\Pi_{12}=\frac{1}{2}I$. For this choice of $\Pi$, if~(\ref{eq:DeltaAntidissipat}) holds for valid $Q$ and $\varsigma$ then, by choosing $\zeta=\dot{x}$, we conclude that the PDM is $\delta$-antipassive according to~\cite[Definition~12]{Park2019From-population}. In particular, (\ref{PDMdissip}) implies~\cite[(54b)]{Park2019From-population}, while (\ref{vanish1}) and (\ref{ssIO}) imply~\cite[(54a)]{Park2019From-population}. Note that the opposite may not hold unless one can show the existence of an appropriate $\varsigma$ satisfying~(\ref{PDMdissip}) and~(\ref{vanish2}). 
\end{remark}

We proceed now with a proof of Theorem~\ref{PDMthm}.
\vspace{-.15in}
\begin{proof}
 The rest points are the solutions of the simultaneous equations $\nu(x,p)=0$, $f(q,x)=0$ and $p=h(q,x)$.  Since the latter two imply $p=F(x)$ by (\ref{ssIO}),  Nash stationarity property (\ref{NS}) of the EDM  ensures that the $x$ components of the rest points are Nash equilibria. To prove global asymptotic stability, we 
 use the Lyapunov function $$V(x,q)=S(x,h(q,x))+Q(q,x),$$
which is nonnegative definite and, from (\ref{vanish1}) and (\ref{informative}), vanishes on the set of rest points.
Note that
 \begin{eqnarray*}
 \frac{\partial V(x,q)}{\partial x} \!\!\! &=& \!\!\! \frac{\partial Q(q,x)}{\partial x}+ \left.\frac{\partial S(x,p)}{\partial x}\right|_{p=h(q,x)} \\
 &&+ \left.\frac{\partial S(x,p)}{\partial p}\right|_{p=h(q,x)} \frac{\partial h(q,x)}{\partial x} \\
  \frac{\partial V(x,q)}{\partial q} \!\!\! &=& \!\!\! \left.\frac{\partial S(x,p)}{\partial p}\right|_{p=h(q,x)} \frac{\partial h(q,x)}{\partial q}+ \frac{\partial Q(q,x)}{\partial q}.
 \end{eqnarray*}
 Then, we write
 \begin{eqnarray}\nonumber 
 && \!\!\!\!\!\!\!\!  \!\!\!\!\!\!\!\! \!\!\!\!\!\!\!\!\frac{\partial V(x,q)}{\partial x}\nu(x,h(q,x))+ \frac{\partial V(x,q)}{\partial q}f(q,x) \\
 &&\qquad \qquad \qquad \quad  =\Theta_1(x,q)+\Theta_2(q,x)
\label{Q1Q2}
 \end{eqnarray}
where
 \begin{eqnarray*}
&& \!\!\!\!\!\!\!\! \!\!\!\!  \Theta_1(q,x) := \frac{\partial Q(q,x)}{\partial q}f(q,x)+\frac{\partial Q(q,x)}{\partial x}\nu(x,h(q,x))\\
&& \!\!\!\!\!\!\!\! \!\!\!\! \Theta_2(q,x) := \left.\frac{\partial S(x,p)}{\partial x}\right|_{p=h(q,x)}\nu(x,h(q,x))\\
&& \qquad \quad +\left.\frac{\partial S(x,p)}{\partial p}\right|_{p=h(q,x)} \frac{\partial h(q,x)}{\partial x}\nu(x,h(q,x))\\
&& \qquad \quad + \left.\frac{\partial S(x,p)}{\partial p}\right|_{p=h(q,x)} \frac{\partial h(q,x)}{\partial q}f(q,x).
\end{eqnarray*}
Next we note from (\ref{PDMdissip}) and (\ref{dissipation}) that
 \begin{eqnarray*}
 \Theta_1(q,x) &\le& -\varsigma(q,x) -\begin{bmatrix}u\\ \zeta \end{bmatrix}^T \Pi \begin{bmatrix}u  \\ \zeta \end{bmatrix} 
 \\
\Theta_2(q,x) &\le& -\sigma(x,h(q,x))+\begin{bmatrix} u \\ \zeta \end{bmatrix}^T \Pi \begin{bmatrix} u \\ \zeta \end{bmatrix}
\end{eqnarray*}
where $\zeta:=\nu(x,h(q,x))$ and  $u:=\frac{\partial h(q,x)}{\partial x}\zeta
+\frac{\partial h(q,x)}{\partial q}f(q,x)$. Substituting in (\ref{Q1Q2}), we get 
 \begin{eqnarray} \nonumber 
 &&\!\!\!\!\!\!\!\! \!\!\!\!  \!\!\!\!  \frac{\partial V(x,q)}{\partial x}\nu(x,h(q,x))+ \frac{\partial V(x,q)}{\partial q}f(q,x) \\
 &&\qquad \qquad  \le-\varsigma(q,x) -\sigma(x,h(q,x)),
 \end{eqnarray}
where the right-hand side is negative semidefinite and, from (\ref{vanish2}) and (\ref{negdef}),
vanishes on the set of rest points. Thus, we conclude global asymptotic stability of this set.
\end{proof}

\vspace{-.1in}
\subsection{Example: Congestion Game with Mixed Autonomy and Smoothing Dynamics}
We proceed to analyze the following dynamical version of Example~\ref{mixed} in which the payoff responds to changes in $F\big(x(t)\big)$ according to a first-order system that, as argued in~\cite{Fox2013Population-Game} for a similar example, smooths short-term fluctuations and isolates longer term trends. The smoothing dynamics can account for, {\it e.g.}, the time lag with which the drivers receive and process congestion information.

\begin{definition} \label{def:dynamicmixed} {\bf (Congestion Game with Mixed Autonomy and Smoothing Dynamics) } Given the link delay functions $\Phi_\ell$, $\ell=1,\dots,L$, and the routing matrix $R$  in Example~\ref{mixed}, and a positive time constant $\tau$, the  PDM to be analyzed here is:

\begin{subequations}
\label{eq:dyn1mixeddyn}
\begin{eqnarray} \label{eq:dyn1mixeddyn-a}
\tau \dot{q}(t)&=& - q(t)+ \begin{bmatrix} \Phi_1(z_1(t)) \\ \vdots \\ \Phi_L(z_L(t)) \end{bmatrix}  \\
p(t)&=&-\begin{bmatrix}R\\ R\end{bmatrix}q(t),
\end{eqnarray}
\end{subequations} where $z$ is as defined in (\ref{zdef}), 
 $q(t)\in \mathbb{R}^L$, and $p(t)\in \mathbb{R}^{2N}$. Note that  (\ref{eq:dyn1mixeddyn}), with steady-state condition $\dot{q}(t)=0$, recovers  (\ref{payoff-vec}); that is, (\ref{ssIO}) holds.

\end{definition}

\begin{assumption} \label{assu:PDM} {\bf (Delay monotonicity)} Recall that the function $\Phi_\ell$ in Example~\ref{mixed} represents the delay on link $\ell$ as a function of  $z_\ell\ge 0$. We assume that this function is {\it strictly} increasing, continuously differentiable, and surjective with codomain $[\alpha_\ell,\infty)$. Since $\Phi_\ell(z_\ell)$ is in $[\alpha_\ell,\infty)$,
the set $[\alpha_1,\infty)\times \cdots \times [\alpha_L,\infty)$ is forward invariant for $q(t)$ in (\ref{eq:dyn1mixeddyn-a}). 
\end{assumption}

When $\Phi_\ell$ satisfies Assumption~\ref{assu:PDM}, we let $\phi_1,\cdots,\phi_L$ be functions such that
 $
 \phi'_\ell(z_\ell)=\Phi_\ell(z_\ell)
 $
and define
$$
\phi(z):=\phi_1(z_1)+\cdots+\phi_L(z_L),
$$
which is  strictly convex since $\Phi_\ell$, $\ell=1,\dots,L$, are strictly increasing functions.
It follows that
$
\nabla\phi(z)=\begin{bmatrix} \Phi_1(z_1) & \cdots & \Phi_L(z_L) \end{bmatrix}^T
$
and (\ref{eq:dyn1mixeddyn-a}) can be rewritten as
$$
\tau \dot{q}(t)= - q(t)+\nabla \phi(z(t))
$$
or, equivalently, as in (\ref{PDM1}) with
\begin{equation}\label{fqx}
f(q,x)=\left.\frac{1}{\tau}(-q+\nabla \phi(z))\right|_{z=[\mu R^T \ R^T ]x}.
\end{equation} 

\subsubsection{Storage function based on Legendre's transform}

To ascertain the stability of~(\ref{eq:meandynamic}) when $F$ is contractive and the revision protocol of the EDM is of the perturbed best response (PBR) type, \cite[Theorem~3.1]{Hofbauer2007Evolution-in-ga} constructed a Lyapunov function that makes use of Legendre's transform.  In~\cite[Section~IX.B]{Park2018Payoff-Dynamic-}, this idea was adapted to construct a storage function for a smoothing PDM based on a potential game $F$, subject to additional constraints on the image of $F$. While the Lyapunov function in~\cite[Theorem~3.1]{Hofbauer2007Evolution-in-ga} involved the Legendre transform of the so-called {\it deterministic perturbation}, the storage function in~\cite[Section~IX.B]{Park2018Payoff-Dynamic-} incorporates the Legendre transform of the potential of $F$. Subject to Assumption~\ref{assu:PDM} and motivated by the latter approach, we propose the following candidate storage function for the PDM~(\ref{eq:dyn1mixeddyn}):
\begin{equation}\label{Qdef}
Q(q,x):=\left.\frac{1}{\tau}(\phi(z)-q^Tz-\phi^*(q))\right|_{z=[\mu R^T \ R^T]x}
\end{equation} with $q \in [\alpha_1,\infty)\times \cdots \times [\alpha_L,\infty)$ and $x \in \mathbb{R}^{2N}_{\ge 0}$. Here, $\phi^*$ is the Legendre transform of $\phi$, defined as
\begin{equation}\label{legendre}
\phi^*(q) : = \min_{y \in \mathbb{R}_{\ge 0}^{L}}  \{\phi(y) - q^Ty\}.
\end{equation}
Note that $\phi^*(q)$ is well-defined for $q\in [\alpha_1,\infty)\times \cdots \times [\alpha_L,\infty)$, since the minimization in (\ref{legendre}) decomposes
into $$\min_{y_\ell\ge 0}\phi_\ell(y_\ell)-q_\ell y_\ell, \quad \ell=1,\dots,L,$$
which has  unique solution satisfying $\phi'_\ell(y_\ell)=\Phi_\ell(y_\ell)=q_\ell$ when $q_\ell\in [\alpha_\ell,\infty)$ by surjectivity of $\Phi_\ell$. 

\begin{remark} {\bf ($F$ is not a potential game)}  The approach in~\cite[Section~IX.B]{Park2018Payoff-Dynamic-} applies only to potential games~\cite{Monderer1996Potential-games,Sandholm2001Potential-games}. However,  $F$ in Example~\ref{mixed} is not a potential game, since its Jacobian is asymmetric. We circumvent this difficulty by working with the Legendre transform of $\phi$.
\end{remark}

The following proposition guarantees that this choice for $Q$ satisfies the conditions of Theorem~\ref{PDMthm} for $\Pi$ as chosen in Lemma~\ref{multi-passive}, and $W$ is as in (\ref{W}). In contrast to \cite[Proposition~9]{Park2018Payoff-Dynamic-}, no other requirements, other than Assumption~\ref{assu:PDM}, need to be imposed on the image of $F$. In addition, our proof for the proposition guarantees the existence of an associated $\zeta$
for which (\ref{PDMdissip}),(\ref{vanish1}),(\ref{vanish2}) hold.

\begin{proposition}\label{prop:dynamicmixed} 
If Assumption~\ref{assu:PDM} holds then $Q$, as defined in~(\ref{Qdef}) for the PDM~(\ref{eq:dyn1mixeddyn}), satisfies the hypotheses (\ref{PDMdissip}),(\ref{vanish1}),(\ref{vanish2})
of Theorem \ref{PDMthm} with
 \begin{equation}\label{Pi4ex}
 \Pi_{11}=\Pi_{22}=0 \quad \mbox{and} \quad \Pi_{12}=\Pi_{21}=\frac{1}{2}W,
 \end{equation}
 where $W$ is as in (\ref{W}).
\end{proposition}
 
 \begin{proof} It follows from (\ref{Qdef})-(\ref{legendre}) that $Q(q,x)\ge 0$, and it vanishes when $z=[\mu R^T \ R^T]x$ is the  minimizer in (\ref{legendre}), that is when
$
\nabla\phi(z)=q.
$
Since this is the same condition for $f(q,x)$ in (\ref{fqx}) to vanish, we conclude that (\ref{vanish1}) holds. 
Next note that
\begin{subequations}\label{Qx}
\begin{eqnarray}
\frac{\partial Q(q,x)}{\partial x}\zeta \!\!&=&\!\!\frac{1}{\tau}(-q+\nabla \phi(z))^T
\begin{bmatrix} \mu R^T & R^T \end{bmatrix}\zeta\\
\!\!&=&\!\!f(q,x)^T \label{Qxpartb}
\begin{bmatrix} \mu R^T & R^T \end{bmatrix}\zeta \\
\!\!&=&\!\!-\psi^T \Pi \psi. \label{Qxpartc}
\end{eqnarray}
\end{subequations}
Here (\ref{Qxpartc}) follows because, from (\ref{psi-def}) and $h(q,x)=-\begin{bmatrix}R \\ R \end{bmatrix}q$
we have
$$
\psi^T=\begin{bmatrix} -f(q,x)^TR^T & -f(q,x)^TR^T & \zeta^T \end{bmatrix}
$$
and, with $\Pi$ as defined in (\ref{Pi4ex}) and (\ref{W}), we get
$$
-\psi^T\Pi\psi=\begin{bmatrix} f(q,x)^TR^T & f(q,x)^TR^T  \end{bmatrix}  \begin{bmatrix} \mu I & 0 \\ 0 & I \end{bmatrix}\zeta,
$$
which is equal to (\ref{Qxpartb}). Thus, (\ref{PDMdissip})
 follows from (\ref{Qx})  with
\begin{equation}\label{varsigdef}
\varsigma(q,x):=-\frac{\partial Q(q,x)}{\partial q}f(q,x).
\end{equation}
To show  (\ref{vanish2}) we note that
\begin{equation}
\frac{\partial Q(q,x)}{\partial q}=-\frac{1}{\tau}\left(\nabla \phi^*(q)+z\right)^T
\end{equation}
and, from (\ref{fqx}) and (\ref{varsigdef}),
\begin{equation}\label{varsigma1}
\varsigma(q,x)=\frac{1}{\tau^2}\left(\nabla \phi^*(q)+z\right)^T(-q+\nabla \phi(z)).
\end{equation}
We define the new variable $\bar{z}=-\nabla\phi^*(q)$, which satisfies the inverse relation\footnote{To see this, let $\bar{z}$ be the minimizer in (\ref{legendre}), that is $\nabla\phi(\bar{z})=q$ and $\phi^*(q)=\phi(\bar{z})-q^T\bar{z}$. Then $\phi(\bar{z})-\hat{q}^T\bar{z}\ge \phi^*(\hat{q})$, with equality when $\hat{q}=q$. Thus, $\phi(\bar{z})=\max_{\hat{q}}\phi^*(\hat{q})+\hat{q}^T\bar{z}$ and $\hat{q}=q$ is the maximizer: $\nabla\phi^*({q})=-\bar{z}$.} $q=\nabla\phi(\bar{z})$ and  rewrite (\ref{varsigma1}) as
\begin{equation}\label{varsigma2}
\varsigma(q,x)=\frac{1}{\tau^2}\left(z-\bar{z}\right)^T(\nabla \phi(z)-\nabla\phi(\bar{z})).
\end{equation}
Since $\phi$ is strictly convex, the expression in (\ref{varsigma2}) is nonnegative and vanishes only when $z=\bar{z}$, that is only when
$
\nabla \phi(z)=\nabla\phi(\bar{z})=q.
$
This is the same condition for $f(q,x)$ in (\ref{fqx}) to vanish, thus (\ref{vanish2}) follows. 
\end{proof}

This proposition allows us to use Theorem~\ref{PDMthm} to conclude global asymptotic stability  for  the set of rest points of (\ref{eq:MeanClosedLoopModel}) when the EDM satisfies the conditions of Lemma~\ref{multi-passive} and the PDM is specified by~(\ref{eq:dyn1mixeddyn}), subject to Assumption~\ref{assu:PDM}. In particular, we can use Remark~\ref{rem:IPCIsEachPopDiss} to state without proof the following counterpart of Corollary~\ref{cor:stabilityExample1}:

\begin{corollary}  If the PDM~(\ref{eq:dyn1mixeddyn}) satisfies Assumption~\ref{assu:PDM}, and the EDM is of the IPC class, such as the Smith EDM, then the set of rest points of (\ref{eq:MeanClosedLoopModel}) (whose $x$ components constitute $NE(F)$) is globally asymptotically stable.
\end{corollary}



\section{Conclusion}
We presented dissipativity tools to establish global asymptotic stability of the set of Nash equilibria in a deterministic model of population games. This model allows for a dynamic payoff mechanism as well as broad classes of protocols by which the agents revise their strategies. Our results generalized those in \cite{Hofbauer2009Stable-games-an} that use contraction properties of the payoff, and those in \cite{Fox2013Population-Game,Park2018Passivity-and-e,Park2019From-population,Park2018Payoff-Dynamic-} that relate contraction to passivity properties and account for dynamical payoff mechanisms. We defined the notion of $\delta$-dissipativity for the evolutionary dynamics model and showed in Theorems \ref{main} and \ref{PDMthm} that a complementary property of the payoff dynamics model guarantees stability for the set of Nash equilibria. This complementary property  relaxed the contraction property of the payoff used in earlier results and allowed us to define a class of ‘weighted contraction’ games, which encompasses a routing game example with mixed autonomy.  We also presented a numerical method that uses convex optimization to check the aforementioned relaxed contraction properties. We hope that the results of this paper will enable researchers to study broader classes of payoff dynamics arising in applications.

\vspace{-.1in}
\begin{appendices}
\section{Proof of Proposition \ref{prop2}}
Since  $J(x)\in {\rm box}\{ G_0, C_1D_1^T,\dots, C_dD_d^T \}$  and $\Pi_{11}=0$ we will prove (\ref{incremental}) by showing that
\begin{eqnarray} \label{qua}
&&\!\!\!\!\! \!\!\!\!\! \!\!\!  P(\Pi_{12}^T(G_0+ \gamma_1C_1D_1^T+\cdots +\gamma_dC_dD_d^T)+\\
&&\!\!\!\!\! (G_0^T+ \gamma_1D_1C_1^T+\cdots +\gamma_dD_dC_d^T)\Pi_{12}+\Pi_{22})P\le 0 \nonumber 
\end{eqnarray}
for all $\gamma_i \in [0,1]$. 
Inequality (\ref{qua}) means that, for all $x\in \mathbb{R}^n$,
\begin{eqnarray}\label{quad}
&&\!\!\!\!\! \!\!\!\!\! \!\!\!  \!\!\!\!\! \!\!\!\!\! \!\!\! x^TP(\Pi_{12}^TG_0+G_0^T\Pi_{12}+\Pi_{22})Px \\ \nonumber
&& \!\!\!\!\! \!\!\!\!\! \!\!\!  +x^TP\Pi_{12}^T(C_1y_1+\cdots +C_dy_d) \\  \nonumber
&& +(C_1y_1+\cdots +C_dy_d)^T\Pi_{12}Px\le 0,
\end{eqnarray}
where we have substituted $y_i:=\gamma_iD_i^TPx$. Since $\gamma_i \in [0,1]$, we have 
\begin{equation}\label{xycon}
y_i^TD_i^TPx \ge y_i^Ty_i \quad i=1,\dots,d.
\end{equation}
Thus, we wish to show that (\ref{quad}), rewritten here as
\begin{equation}\label{quad1}
\begin{bmatrix} x \\ y_1 \\ \vdots \\ y_d \end{bmatrix}^T
\begin{bmatrix} P(\Pi_{12}^TG_0+G_0^T\Pi_{12}+\Pi_{22})P & P\Pi_{12}^TC  \\ C^T\Pi_{12}P & 0\end{bmatrix} 
\begin{bmatrix} x \\ y_1 \\ \vdots \\ y_d \end{bmatrix} \le 0,
\end{equation}
holds when the variables $x,y_1,\cdots,y_d$ are constrained by (\ref{xycon}) or, equivalently,
\begin{equation}\label{quad2}
\begin{bmatrix} x \\ y_1 \\ \vdots \\ y_d \end{bmatrix}^T
\begin{bmatrix}0 & PD_iE_i  \\ E_i^TD_i^TP & -2E_i^TE_i\end{bmatrix} 
\begin{bmatrix} x \\ y_1 \\ \vdots \\ y_d \end{bmatrix} \ge 0 \quad i=1,\dots,d,
\end{equation}
where $E_i$ is the $\varrho_i \times (\varrho_1+\cdots +\varrho_d)$ matrix such that 
$$
\begin{bmatrix}  E_1 \\ \vdots \\ E_d \end{bmatrix}= I_{\varrho_1+\cdots +\varrho_d}, \quad \mbox{therefore} \quad E_i\begin{bmatrix}  y_1 \\ \vdots \\ y_d \end{bmatrix}=y_i.
$$
It follows from the S-procedure \cite{LMI} that if there exist nonnegative constants $\omega_1,\cdots, \omega_d$ such that
\begin{eqnarray}\nonumber
&& \!\!\!\!\! \!\!\!\!\! \!\!\!  \!\!\!\!\! \!\!\!\!\! \!\!\!  \begin{bmatrix} P(\Pi_{12}^TG_0+G_0^T\Pi_{12}+\Pi_{22})P & P\Pi_{12}^TC  \\ C^T\Pi_{12}P & 0\end{bmatrix}  \\
&& +\sum_{i=1}^d \omega_i \begin{bmatrix}0 & PD_iE_i  \\ E_i^TD_i^TP & -2E_i^TE_i\end{bmatrix} \le 0 \label{Sproc}
\end{eqnarray}
then  (\ref{quad1}) holds whenever (\ref{quad2}) does. Indeed  (\ref{Sproc}) can be rewritten as (\ref{Scheck}), and the existence of nonnegative constants $\omega_1,\cdots, \omega_d$ is the hypothesis of the Proposition. Thus, (\ref{quad1}) holds whenever (\ref{quad2}) does, we conclude (\ref{qua}) and, hence, (\ref{incremental}).

\end{appendices}

\end{document}